\documentclass[letterpaper, 11 pt, peerreviewca]{ieeeconf}

\usepackage{courier}
\usepackage{hyperref}
\usepackage{graphicx,subfigure}
\usepackage{epstopdf}
\usepackage{tikz}
\usepackage{amsmath}
\usepackage{multirow}
\usepackage{amsfonts}
\usepackage{float}
\usepackage{amssymb}
\usepackage{graphicx}
\usepackage{flushend}
\usepackage[ruled]{algorithm}
\usepackage{algpseudocode}
\newtheorem{theorem}{Theorem}
\newtheorem{lemma}{Lemma}
\newtheorem{proposition}{Proposition}

\newtheorem{definition}{Definition}
\newtheorem{remark}{Remark}
\newtheorem{example}{Example}
\makeatletter

\newcommand{\Rmnum}[1]{\expandafter\@slowromancap\romannumeral #1@}

\newenvironment{breakablealgorithm}
{
	\begin{center}
		\refstepcounter{algorithm}
		\hrule height.8pt depth0pt \kern2pt
		\renewcommand{\caption}[2][\relax]{
			{\raggedright\textbf{\ALG@name~\thealgorithm} ##2\par}%
			\ifx\relax##1\relax 
			\addcontentsline{loa}{algorithm}{\protect\numberline{\thealgorithm}##2}%
			\else 
			\addcontentsline{loa}{algorithm}{\protect\numberline{\thealgorithm}##1}%
			\fi
			\kern2pt\hrule\kern2pt
		}
	}{
		\kern2pt\hrule\relax
	\end{center}
}

\makeatother

\begin{document}
\title{\LARGE Neuromimetic Linear Systems --- Resilience and Learning}
\author{Zexin Sun \& John Baillieul}
\maketitle
\let\thefootnote\relax\footnotetext{\noindent
\hspace{-0.1in}\hrulefill
\hspace{0.8in}\\John Baillieul is with the Departments of Mechanical Engineering, Electrical and Computer Engineering, and the Division of Systems Engineering at Boston University, Boston, MA 02115. Zexin Sun is with the Division of Systems Engineering at Boston University.  The authors may be reached at {\tt \{johnb, zxsun\}@bu.edu}. \newline Support from various sources including the Office of Naval Research grant number N00014-19-1-2571 is gratefully acknowledged. }


\begin{abstract}
\noindent Building on our recent work on {\em neuromimetic control theory}, new results on resilience and neuro-inspired quantization are reported.  The term neuromimetic refers to the models having features that are characteristic of the neurobiology of biological motor control.  As in previous work, the focus is on what we call {\em overcomplete} linear systems that are characterized by larger numbers of input and output channels than the dimensions of the state.  The specific contributions of the present paper include a proposed {\em resilient} observer whose operation tolerates output channel intermittency and even complete dropouts.  Tying these ideas together with our previous work on resilient stability, a resilient separation principle is established.  We also propose a {\em principled quantization} in which control signals are encoded as simple discrete inputs which act collectively through the many channels of input that are the hallmarks of the overcomplete models.  Aligned with the neuromimetic paradigm, an {\em emulation} problem is proposed and this in turn defines an optimal quantization problem.  Several possible solutions are discussed including direct combinatorial optimization, a Hebbian-like iterative learning algorithm, and a deep Q-learning (DQN) approach.  For the problems being considered, machine learning approaches to optimization provide valuable insights regarding comparisons between optimal and nearby suboptimal solutions.  These are useful in understanding the kinds of resilience to intermittency and channel dropouts that were earlier demonstrated for continuous systems.
\end{abstract}
\begin{flushleft} {{\bf Keywords}: Neuromimetic control, parallel quantized actuation, channel intermittency, neural emulation} 
\end{flushleft}

 \section{Introduction}

Recent research has been aimed at understanding control theoretic models in which the numbers of control input channels or system observation channels exceed -- and perhaps greatly exceed -- the dimension of the state \cite{Baillieul1},\cite{Baillieul2},\cite{baillieul2021}.  Because these control system models may be such that they remain controllable and observable even though some of the control or output channels are removed, we shall call such systems {\em overcomplete}, borrowing terminology from signal processing.
There are a number of important features that may be found in linear models of this type, including the fact that by adding control channels it is possible to reduce the control-energy cost of moving between arbitrary endpoints in the state space and the effect of noise can be reduced.  It was shown in \cite{Baillieul2} that having more control input channels made it possible to design constant gain feedback laws that produced resilient stability in the sense that the closed-loop system remained asymptotically stable despite control channels being only intermittently available or even permanently unavailable.  The present paper extends this work on resilient stability and control in several directions.  Complementing results on feedback control designs that are resilient to channel drop-outs, we propose a class of observer designs that are similarly resilient with respect to channel dropouts.  With these observers, we show that there is a corresponding resilient separation principle such that feedback controls that use the observer-based estimates perform well even when there may be both output and control channel dropouts.

\smallskip

Beyond the resilience advantage of overcomplete systems, a long standing aim of our research has been to explore the performance of such systems when their operation is governed by inputs and output that are drawn from finite dictionaries.  The second part of the paper treats an {\em emulation} problem in which the goal is to design control protocols that select control actions from a finite set in such a way that the system with quantized inputs will closely emulate the performance of a prescribed closed-loop linear system.  Because a long-term goal is an understanding of online learning and adaptation, our approach addresses the emulation problem using two different learning algorithms---a simple Hebb-Oja algorithm \cite{oja1982simplified} and our own version of a recently proposed deep Q-learning algorithm \cite{mnih2015human}.  Some low dimensional examples show that these learning methods can be resilient in the sense that the system can overcome the loss of an input channel and relearn good emulation protocols using the same dictionaries.  The paper concludes with a brief discussion of planned research whose aim is to further explore the space of quantized control with very large numbers of input and output channels.

\section{Linear Systems with Large Numbers of Inputs and Outputs}

The linear time-invariant (LTI) systerms to be studied have the simple form
\begin{equation}
\begin{array}{l}
\dot x(t)=Ax(t) + Bu(t), \ \ \ x\in\mathbb{R}^n, \ \ u\in\mathbb{R}^m, \ {\rm and}\\[0.07in]
y(t)=Cx(t), \ \ \ \ \ \ \ \ \ \ \ \ \ \ y\in\mathbb{R}^q.\end{array}
\label{eq:jb:linear}
\end{equation}
The primary way that the models of this form in what follows are different from most linear time-invariant control systems in the literature is that we shall assume large numbers of input channels ($m\gg n$) and large numbers of output channels ($q\gg n$).

\bigskip

\subsection{Lifted operators associated with large multiplicities of input/output channels}
Under these assumptions on dimensionality of the state, input, and output spaces, the matrices $B:\mathbb{R}^m\rightarrow\mathbb{R}^n$ and $C:\mathbb{R}^n\rightarrow\mathbb{R}^q$ define respectively a left action on $m\times n$ matrices $\mathbb{R}^{m\times n}$ and a right action on $\mathbb{R}^{n\times q}$: ${\cal L}_B({\rm U})=B\cdot {\rm U}$ and ${\cal R}_C({\rm U})={\rm U}\cdot C$.  These {\em liftings} are depicted by
\[
\begin{array}{clcc}
{\cal L}_B:&\mathbb{R}^{m\times n}&\rightarrow&\mathbb{R}^{n\times n}\\
&\ \ \big\vert&&\big\vert\\
B:&\mathbb{R}^m&\rightarrow&\mathbb{R}^n
\end{array}
\ \ \ \ {\rm and} \ \ \ \ 
\begin{array}{clcc}
{\cal R}_C:&\mathbb{R}^{n\times q}&\rightarrow&\mathbb{R}^{n\times n}\\
&\ \ \big\vert&&\big\vert\\
C:&\mathbb{R}^n&\rightarrow &\mathbb{R}^q.\end{array}
\]
respectively.  The following preliminary observations will be useful in characterizing stable constant feedback and observer gains that are resilient with respect to channel dropouts and intermittency.

\begin{lemma} Suppose $m,q>n$ and $B:\mathbb{R}^m\rightarrow\mathbb{R}^n$ and $C:\mathbb{R}^n\rightarrow\mathbb{R}^q$ have full rank $n$.  Then ${\cal L}_B$ and ${\cal R}_C$ have rank $n^2$ and nullspace dimensions $n(m-n)$ and $(q-n)n$ respectively.
\end{lemma}

\begin{proof}
The proof of the statement for ${\cal L}_B$ may be found in \cite{baillieul2021} (Lemmas 3,4), and the proof in the case of ${\cal R}_C$ is a simple modification.
\end{proof}

\bigskip

\begin{lemma} Let the matrices $A,B,C$, appearing in (\ref{eq:jb:linear}) be $n\times n$, $n\times m$, and $q\times n$ respectively with $m,q>n$ and with $B$ and $C$ having full rank $n$.  Then there is an $n(m-n)$ parameter family of solutions $X=\hat A$ to the matrix equation ${\cal L}_B(X)=A$ and a $(q-n)n$ parameter family of solutions $X=\hat A$ to the matrix equation ${\cal R}_C(X)=A$.
\end{lemma}

\begin{proof}
The proof in the case of ${\cal L}_B$ was given in \cite{baillieul2021}.  Let ${\cal N}({\cal R}_C)$ denote the $(q-n)n$-dimensional nullspace of ${\cal R}_C$.  A particular solution to the matrix equation is $A(C^TC)^{-1}C^T$, and any other solution may be written as $A(C^TC)^{-1}C^T+N$, where $N\in{\cal N}({\cal R}_C)$.
\end{proof}


\bigskip
\bigskip

Under the stated assumptions on the dimensions of $A$, $B$, and $C$ in (\ref{eq:jb:linear}), we consider an identity observer of the form
\begin{equation}
\label{eq:jb:observer}
\begin{array}{l}
\dot z(t)=Az(t) + E(y(t) - Cz(t)) + Bu(t),  \ \ \ z\in\mathbb{R}^n\\[0.07in]
\hat y(t)=Cz(t).\end{array}
\end{equation}
To study the effect of observation channel intermittency and dropouts, we introduce the following
notation. Let $q$ be a positive integer and let $[q]$ denote the set of integers $\{1,\dots,q\}$.  Let $[q]\choose k$ be the set of $k$-element subsets of [q], for instance, ${[3]\choose 2} = \left\{\{1,2\},\{1,3\},\{2,3\}\right\}$.  With $I\subset {[q]\choose k}$, Let $P_I$ be a diagonal matrix whose diagonal entries are $0$'s and $1$'s with $1$'s in the positions $I$.  
The goal is to find an observer gain $E$ such that $A$--$EP_IC$ is Hurwitz for all $n\le j\le q$ and $I\in {[q]\choose j}$.  While this goal cannot generally be met for all values of $j$ in the interval $n\le j \le q$ and $I$ ranging over all of  ${[q]\choose j}$, the next result gives conditions under which the closed observer loop matrix $A$--$EP_IC$ is Hurwitz.

\medskip

\begin{definition}\rm
Let $I\in{[q]\choose j}$.  The {\it lattice of index supersets of I} in $[m]$ is given by $\Lambda_I=\{L\subset [q]: I\subset L\}$.
\end{definition}

\medskip

\begin{theorem}
Consider the LTI system (\ref{eq:jb:linear}) and identity observer (\ref{eq:jb:observer}), where the continuing assumption on dimensions is that $m,q,>n$.  Suppose all principal minors of $C$ are nonzero, that $X=\hat A$ is a solution of ${\cal R}(X)=A$, and that for a given $I\in {[q]\choose j}$, $\hat A$ is {\em right-invariant} under $P_I$---i.e. $\hat AP_I=\hat A$.  Then for any real $\alpha>0$ and observer gain $E=\alpha C^T+ \hat A$, the observer error $e=z-x$ with observation channel dropouts given by $y=P_LCx$ satisfies $\dot e = (A-EP_LC) e$ whose solution approaches $0$ exponentially asymptotically for all $L\in\Lambda_I$.
\end{theorem}

\medskip

\noindent The proof will make use of the following.

\begin{lemma}
Consider the LTI system (\ref{eq:jb:linear}) where $A$ is an $n\times n$ real matrix and $C$ is a $q\times n$ real matrix with $q>n$.  Suppose that $C$ has full rank $n$ and that $\hat A$ is an $n\times q$ solution of ${\cal R}_C(\hat A)=A$.  Then for any real $\alpha>0$, the matrix
\[
E=\alpha C^T+\hat A
\]
places the eigenvalues of $A-EC$ in the open left half-plane.
\end{lemma}

\smallskip

\begin{proof} 
The following equalities
\[
A-EC = A -(\alpha C^T + \hat A)C = A-\alpha C^TC -A = -\alpha C^TC
\]
show that $A-EC$ is a negative definite matrix.  Hence its eigenvalues are in the left half-plane.
\end{proof}

\smallskip

{\em Proof of Theorem 1:}  If all principal minors of $C$ are nonzero, then $C^TP_L C$ is positive definite for all $L\in\Lambda_I$.  The theorem then follows from Lemma 3 with $P_LC$ and $EP_L$ playing the roles of $C$ and $E$ in the lemma. \ \ \ \ \ \ \ \ \ \ 
{\small $\blacksquare$}

\medskip

\begin{theorem}
Consider the LTI system (\ref{eq:jb:linear}) under the continuing assumption that $m,q>n$.  Suppose that all principal minors of both $B$ and $C$ are nonzero.  Let $X=\hat A_1$ be a solution of ${\cal L}_B (X)=A$ and $Y=\hat A_2$ be a solution of ${\cal R}_C(Y)=A$.  Suppose further 
\begin{enumerate}
\item that for a given $j\ge n$ and $I_1\in{[m]\choose j}$, $P_{I_1}$ is an $m\times m$ diagonal matrix with $1$'s on the diagonal corresponding to the index set $I_1$ and $0$'s elsewhere, and that $\hat A_1$ is left-invariant under $P_{I_1}$ (i.e. $P_{I_1}\hat A_1 = \hat A_1$);
\item that for some (possibly different) $j\ge n$ and $I_2\in{[q]\choose j}$, $P_{I_2}$ is a $q\times q$ diagonal matrix corresponding in the same way to the index set $I_2$, and that $\hat A_2$ is right-invariant under $P_{I_2}$ (i.e. $\hat A_2  P_{I_2}= \hat A_2$).
\end{enumerate}
Then for any $\alpha_1,\alpha_2>0$, the $m\times n$ gain $K=-\alpha_1 B^T-\hat A_1$ and $n\times q$ matrix $E=\alpha_2 C^T+\hat A_2$ render the observer-based control implementation
\begin{equation}
\begin{array}{lll}
\dot x & = & Ax+BP_{L_1}Kz \\
\dot z & = & (A-EP_{L_2}C)z + EP_{L_2} Cy +BP_{L_1}Kz
\end{array}
\label{eq:jb:observer}
\end{equation}
asymptotically stable for all 
\[
\begin{array}{c}
L_1\in\Lambda_{I_1}= \{L\subset[m]:I_1\subset L\}, \\ 
L_2\in\Lambda_{I_2}= \{L\subset[q]:I_2\subset L\}.\end{array}
\]
\label{thm:jb:observer}
\end{theorem}

\begin{proof}
Let $e=z-x$.  The equation (\ref{eq:jb:observer}) may be rewritten as
\begin{equation}
\small
\begin{array}{ll}
\left(\begin{array}{l}
\dot x\\
\dot e\end{array}\right) & = 
\left(\begin{array}{cc}
A+BP_{L_1}K & BP_{L_1}K\\
0 & A-EP_{L_2}C\end{array}\right)
\left(\begin{array}{l}
 x\\
 e\end{array}\right)\\[0.1in]
 & = 
\left(\begin{array}{cc}
B(\hat A_1+ P_{L_1}K) & BP_{L_1}K\\
0 & (\hat A_2 -EP_{L_2})C\end{array}\right)\left(\begin{array}{l}
 x\\
 e\end{array}\right)\\[0.1in]
& =  \left(\begin{array}{cc}
-\alpha_1 BP_{L_1}B^T & BP_{L_1}K\\
0 & -\alpha_2C^TP_{L_2}C\end{array}\right)\left(\begin{array}{l}
 x\\
 e\end{array}\right).
\end{array}
\label{eq:jb:rewrite}
\end{equation}
It is well-known (and easily proven) from linear algebra that the eigenvalues of an upper block triangular matrix are the union of the set of eigenvalues of the diagonal blocks.  Hence, the eigenvalues of the coefficient matrices in (\ref{eq:jb:rewrite}) are the combined eigenvalues of $-\alpha_1 BP_{L_1}B^T$ and $-\alpha_2 C^TP_{L_2}C$.  That these matrices are strictly negative definite follows from the assumption that all principal minors of $B$ and $C$ are nonzero, together with the assumptions on the projections $P_{L_1},P_{L_2}$.  Hence all eigenvalues for all choices of $L_k\in\Lambda_{I_k},  k=1,2$, are in the open left half plane, proving the theorem.
\end{proof}

\bigskip

\section{Emulation problems with vector quantized control and machine learning}
As discussed in \cite{baillieul2021}, a primary motivation for studying what we have called {\em ovecomplete control systems} is to develop a control theory of systems that exhibit the type of resilience that is observed in neurobiology and that involve input and output signals generated by the collective activity of very large numbers of simple elements.  Thus we seek to understand systems where overall function depends on groups of inputs and can be sustained even though some members of the group of inputs fail to operate.  In this context we aim to understand neurobiology-inspired approaches to quantization that will preserve essential features of resilience as characterized in Theorem \ref{thm:jb:observer}.  Our neuro-inspired quantization involves piecewise constant inputs $u$ to (\ref{eq:jb:linear}) taking values in ${\cal U}=\{-1,0,1\}^m$.  While systems with quantized inputs have been studied for many decades ---e.g.\ \cite{Curry},\cite{lewis1963},\cite{liu1992optimal}--- the early work in this area was largely concerned with the effects digital round-off errors on control performance.  Renewed interest in systems with quantized inputs came about in connection with information-based control theory and the so called data-rate theorem that was first articulated around around 1999-2000.  References to much of this work and an overview are provided in \cite{Nair2}.  Most of information-based control has been aimed at characterizing data rates needed in feedback communication loops to achieve classical control objectives such as stability.  The quantizations we consider next are less focused on classical control-theoretic objectives and more concerned with meeting neurobiology-inspired objectives such as learning and emulation, \cite{Colder},\cite{Marder}.  Nevertheless, they may also provide a basis for understanding the data-rate requirements needed for fine tuning feedback responses in classical linear control theory.  The approaches we propose are specifically tailored to overcomplete control systems.

\subsection{Neuro-inspired quantized control of linear systems}
As in the preceding section, we consider $n\times m$ matrices $B$ that modulate inputs $u\in\{-1,0,1\}^m$.  Each such $u$ is called an {\em activation pattern}, and the set of corresponding vectors $Bu\in\mathbb{R}^n$ is call the {\em quantization output alphabet}.  We note that there are $3^m$ possible activation patterns and $\le 3^m$ elements in the corresponding quantization output alphabet.  A specific case in which this inequality is strict arises as follows.
  Let $\vec x_1, . . . \vec x_k$ be $k>n$ unit vectors in the closed positive orthant of $\mathbb{R}^n$ such that no $n$ of them lie in an $n-1$-dimensional subspace of $\mathbb{R}^n$ and such that the standard unit basis vectors $\vec e_i=(0,\dots,1,\dots, 0)$ consisting of zero entries except for a $1$ in the $i$-th place are present in the set.  The columns of $B$ are these $k$ vectors together with their negatives, $-\vec x_i$, $i,=1,\dots, k$.  It is not difficult to see that the set of possible directions $Bu$ that result from $u$ ranging over $u\in\{-1,0,1\}^m$ is greater for systems with such symmetry than it would be if the columns of $B$ were confined to the positive orthant of $\mathbb{R}^n$.  The set of possible directions $Bu$ where $B$ exhibits such symmetry is also greater when we utilize all activation patterns in $\{-1,0,1\}^m$ than it would be with the same $B$ but activation patterns restricted to binary inputs $u\in \{0,1\}^m$.  To fix ideas, for each $n$, let $B$ be the $n\times 2n$ matrix whose columns are the standard basis vectors and their negatives: $\vec e_i,-\vec e_i$.  Fig.\ \ref{fig:jb:activate} illustrates this setup in the simple case $n=2$ and
\begin{equation}
B=\left(\begin{array}{cccc}
1 & 0 & -1 & 0\\
0 & 1 & 0 & -1\end{array}\right).  
\label{eq:jb:simpleB}
\end{equation}

\begin{figure}
\begin{center}
\includegraphics[scale=0.4]{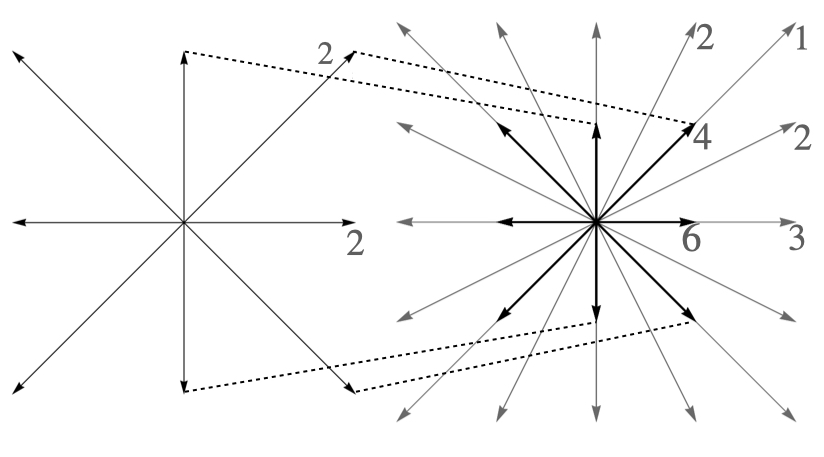}
\end{center}
\caption{The figure illustrates the vectors $Bu$ where $B$ is given by (\ref{eq:jb:simpleB}).  
On the left, the 16 activation patterns are $u\in\{0,1\}^4$ and on the right there are 81 activation patterns $u\in\{-1,0,1\}^4$.  The small numbers indicate the multiplicities of activation patterns that produce each vector direction.  In the figure on the right, there is a greater variety of both vector directions and vector lengths.}
\label{fig:jb:activate}
\end{figure}

We note that in all cases represented in the figure, each vector direction is produced by a multiplicity of activation patterns.  When $u\in\{0,1\}^4$, 16 activation patterns produce eight possible directions, whereas when $u\in\{-1,0,1\}^4$ the 81 activation patterns produce 25 vectors $Bu$.  A simple combinatorial counting argument proves the following.
\medskip
\begin{proposition}
Given a positive integer $n$, define the $n\times 2n$ matrix $B=(\vec e_1,\cdots,\vec e_n, -\vec e_1,\cdots,-\vec e_n)$.  As the $2n$-dimensional activation patterns $u$ range over the $3^{2n}$ element set $\{-1,0,1\}^{2n}$, $5^n$ distinct vector directions are produced.
\end{proposition}

\medskip

The multiplicity of activation patterns that give rise to each vector direction in this proposition grows like $(9/5)^n$.  By adding more input channels such that $m>2n$, the multiplicities of activation patterns giving rise to each $Bu$ grows even faster.  While careful addition of more channels afford the possibility of using piecewise constant quantized control to better approximate continuous systems of the form (\ref{eq:jb:linear}) and also may enhance resilience by virtue of the redundancy in activation patterns, there will at the same time be increased complexity in the algorithms for learning optimal quantized control that will be considered next.

\subsection{Quantized control emulation problems}

As proposed in \cite{baillieul2021}, we formulate the emulation problem of quantized systems as follows.  Let $H$ be an $n\times n$ matrix and consider solutions to the linear ordinary differential equation (ODE) $\dot x = Hx$, $x(0)=x_0$.  The goal of the emulation problem is to find piecewise constant quantized inputs to (\ref{eq:jb:linear}) with sampling interval $h>0$ such that the resulting trajectories of (\ref{eq:jb:linear}) with initial state $x(0)=x_0$ approximate $e^{Hh}x_0$.  The matrix $H$ could be thought of as specifying a goal behavior of a closed loop version of (\ref{eq:jb:linear}), say of the form $H=A+BK$ where $K$ is a stabilizing gain chosen as in Section II.  As outlined in \cite{baillieul2021}, the general problem of emulating trajectories $e^{Ht}x_0$ by trajectories of (\ref{eq:jb:linear}) with quantized inputs may be stated as follows.

\noindent $\bullet$ {\bf General emulation problem} Find a partition of the state space $\{U_i\,:\, \cup\  U_i = \mathbb{R}^n;\ \ U_i^o\cap U_j^o=\emptyset;\ U_i^o={\rm interior}\ U_i\}$ and a selection rule for assigning values  of the input at the $k$-th time step to be $u(k)\in{\cal U}=\{-1,0,1\}^m$, so that for each $x\in U_i$, $Ax + Bu(k)$ is as close as possible to $Hx$ (in an appropriate metric).

Exploiting the properties of linear vector fields, we note that because ode's of the form $\dot x=Hx$ are homogeneous of degree 1, many qualitative features of such systems are determined by evaluating the vector field on the unit sphere in $\mathbb{R}^n$.  In other words, the set of all possible directions taken on by  $Hx$ with $x\in\mathbb{R}^n$ coincides with the set of all directions of $Hx$ with $x\in S^{n-1}$.  Except for scaling, the geometry of the vector field is determined by evaluating it on $S^{n-1}$.   This suggests the following:

\noindent $\bullet$ {\bf Restricted emulation problem} Find a partition of the unit sphere $\{U_i\,:\, \cup\  U_i = S^{n-1}\ \ U_i^o\cap U_j^o=\emptyset;\ U_i^o={\rm interior}\ U_i\}$ and a selection rule $F(x)$ for assigning values $F(x)=u\in{\cal U}=\{-1,0,1\}^m$ such that for each $x\in U_i$ $BF(x)$ is as close as possible to $Hx$ in an appropriate metric.

Formulated in this way, the emulation problems call for metrics.  For the restricted emulation problem, we need to compare vector fields defined on spheres.  Note, that in general the vector fields are neither tangent nor normal; for the given $H$ and sampling interval $h$, they are either of the form $e^{Hh} x_0$ or its first order approximation $(I+Hh)x_0$ with $x_0\in S^{n-1}$.  The quantized system that will  approximate these is $x_0 + hBu$ where $u\in{\cal U}=\{-1,0,1\}^m$.  Here, we are assuming $A=0$. The case $A\ne 0$ and related problems of  stabilization will be treated elsewhere.  The emulation problem is to optimally select activation patterns associated with each $x_0$ so that  $x_0 + hBu$ is ``as close as possible'' to $e^{Hh}\,x_0$, and it remains to define what is meant by ``close''.  Given that we are comparing vector quantities, one possible metric that compares both magnitudes and directions is
\begin{equation}
\left<(Hh)x_0,hBu\right> - wt\,\big|\| Hh x_0\| -\|hBu\|\big|,
\label{eq:jb:composite}
\end{equation}
where $\left< \cdot , \cdot \right>$ denotes the standard Euclidean inner product, and  the weight $wt$ is chosen to reflect the emphasis placed on magnitude vs. direction.  Simple iterative learning may be used to define an appropriate selection rule to optimize such a metric.  To carry this out, the following simple lemma will be useful.
\smallskip

\begin{lemma}
	For $k=1,..,K$, let $G(u_k,x_0)$ be functions taking nonnegative values over a domain of interest. For each $x_0$ in the domain, the index $k = k_{opt}$ of the largest value $G(u_k,x_0)$ may be found by the simple iterative {\bf Algorithm 1}.
\begin{algorithm}
\begin{algorithmic}[1]
	\caption{Iterative learning algorithm for emulation}\label{alg:Iterative_alg}
	
	\State Initialize weights $m_k(0,x_0) = 1$ for $k = 1,2, . . . ,K$ and $\alpha\ge0$;
	
	\Repeat
		\For{$k \gets 1$ to $K$}                    
			\State {$m_k(j+1,x_0)$ =$ \frac{m_k(j,x_0)[\alpha+G(u_k,x_0)]}{\sum_{l=1}^{K}m_l(j,x_0)[\alpha+G_l(x_0)]}$}
		\EndFor
		
	\Until{$j=N$, where for some $k_{opt}$, $m_{k_{opt}}(N,x_0)\gg m_k(N,x_0)$ for all $k\neq k_{opt}$}
	\State Output: activation pattern $k_{opt}$.
\end{algorithmic}
\end{algorithm}
\label{lemma:jb:algorithm}
\end{lemma}

\begin{proof}
Suppose first $\alpha=0$.  Because of the algorithm's initialization, the quantities $m_k(1,x_0)$ are nonnegative and their sum from $k=1$ to $k=K$ is equal to $1$.  Let $G_{k_s}(x_0)$ be the smallest value of the $G(u_k,x_0)$'s.  Let $p=m_{k_s}(1,x_0)$ and $q= \sum_{k\ne k_s}m_k(1,x_0)$.  It follows from 
the iteration formula in the algorithm that $p+q=1$.  We can explicitly evaluate $m_{k_s}(2,x_0)=p^2/(p^2+(1-p)^2)$ and more generally $m_{k_s}(j,x_0)=p^j+{\cal O}(p^{j+1})$.  Noting that if the values $G(u_k,x_0)$ are ordered---say $G_1(x_0)<G(u_2,x_0)<\dots<G(u_K,x_0)$ (i.e.\ $k_s = 1$), then for all $j\ge 1$, $m_1(j,x_0)<\dots<m_K(j,x_0)$.  Under this ordering, let the second smallest element $m_2(j,x_0)=\delta$ where $0\le\delta<1$.  Continuing to iterate, we find that $m_2(j,x_0)\sim \delta^{{\cal O}(j)}+{\cal O}(\delta^{{\cal O}(j)})$.  Hence, $\lim_{j\to\infty} m_2(j,x_0) = 0$, and continuing in this way, it is established that $\lim_{j\to\infty}m_k(j,x_0) = 0$ for all $k<K$.  Because of the renormalization at each step, it also follows that $\lim_{j\to\infty}m_K(j,x_0) = 1$ (i.e.\ the weights associated to the largest value $G(u_k,x_0)$ approach 1).

In the case $\alpha>0$, the stated result remains true, but the speed of convergence of the weights $m_k(j,x_0)$  is significantly reduced.  The case $\alpha=1$ will be of interest in what follows, and for this, one can show that $m_{k_s}(j,x_0)$ is a rational function of $p=m_{k_s}(1,x_0)$ whose numerator grows like $1+(j-1)q+o(q)$ and whose denominator grows like $2^{j-1}+1 +o(q)$.  Details will appear in a longer version of this paper.
\end{proof}

\smallskip

\begin{remark} 
With $\alpha=1$ in 
{\bf Algorithm 1}, the rate of convergence of the weights $m_k(j,x_0)$ as $j\to\infty$ is slower than when $\alpha=0$.  As a consequence, the convergence as $j\to\infty$ of weights associated with the second largest value $G_i(j,x_0)$ will converge to zero more slowly than in the case $\alpha=0$, and much more slowly than the weights $m_k(j,x_0)$ for $k\ne i$.  We omit details of the analysis, but in terms of our quantization problem, it is found that near boundaries of cell $U_i$ in the partition where there is a change in optimal activation patterns, the weight orderings produced by the algorithm anticipate the switch that occurs at the boundary.
\end{remark}

\begin{remark}
The Appendix contains a second algorithm {\bf Algorithm 2}, that optimally selects activation patterns by minimizing a loss function using a neural network and Q-learning.  For the noise-free problems under consideration, the algorithms yield identical results.  The algorithm will be taken up in more detail in an expanded version of this work.
\end{remark}

\section{Quantization Resilience and Complexity}

The approach to quantization in the emulation problem of the preceding section is more complex than the methods used to establish the data-rate theorem, \cite{Nair2}.  Nevertheless, it enables examination of a wider range of the information processing demands of performance requirements that are more exacting than mere stability.  To explore this, we continue with input-modulating matrices $B$ having more columns than rows ($n\times m$ with $m>n$).  As $x$ ranges over $S^{n-1}$, let $F(x)$ be a selection function that assigns as activation pattern to each $x$ in accordance with our learning algorithms..  Let ${\cal V}\subset {\cal U}=\{-1,0,1\}^m$ be the subset of activation patterns that are assigned to at least one point in $S^{n-1}$ by Lemma \ref{lemma:jb:algorithm}.  Because there can be multiple optima associated with any $x\in S^{n-1}$, (as illustrated in Fig.\ \ref{fig:jb:fig2}), the number of possible elements in the {\em quantizaton output alphabet} may be strictly less than the cardinality $|{\cal V}|$.  Let $\{d_1,\dots,d_K\}$ be the direction vectors of this alphabet, and for each $i$, let
\[
D^{-1}(d_i)=\{v\in{\cal V}:Bv=d_i\}\ \ {\rm and}\ \ \alpha_i=|D^{-1}(d_i)|,
\]
the cardinality.

\begin{figure}[ht]
\begin{center}
    \includegraphics[scale=0.27]{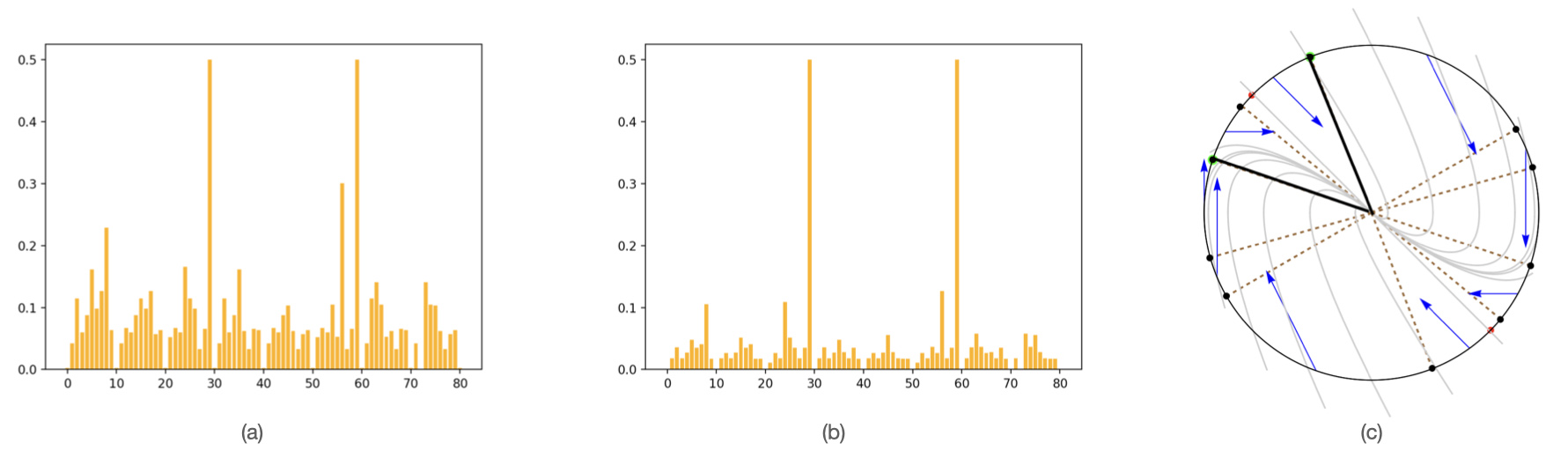}
\caption{We use Algorithms 1 and 2 to learn optimal activation patterns for quantized approximation of $Hhx_0$ where
$H=\left(\begin{array}{cc}
0 & 1\\
-1 & -2\end{array}\right)$
in $\mathbb{R}^2$.  Figures (a),(b) show successive magnitudes of the value functions for a particular point (
$x=(0,1)\in\mathbb{R}^2$.  The two peaks indicate that the algorithms are learning that two activation patterns ($u=(1,-1,0,1),(0,-1,-1,1)$) corresponding to output directions $Bu=(1,-2)$ are 
equally good.  Panel (c) shows the trajectories of the linear vector field (gray) and the ten pie-shaped regions where different activation patterns and quantization directions corresponding to $B$ in (\ref{eq:jb:simpleB}) (Table I) are active.
}
\end{center}
\label{fig:jb:fig2}
\end{figure}

We assume that the selection function $F(\cdot)$ selects activation patterns in a consistent way so that as $x$ ranges over $S^{n-1}$, $F$ selects only one element from each $D^{-1}(d_i)$, $i=1,\dots,K$.  Given such an $F$, we define the partition $U_i=\{x\in S^{n-1}:BF(x)=d_i\}$.  In terms of this partition and a refinement to be introduced next we address the complexity of solutions to the {\em restricted emulation problem}.  Taking inspiration from Shannon's early work on entropy and redundancy in written language, \cite{shannon1951prediction}, we offer the following.  Recalling that the geometric {\em content} (generalized surface area) of the unit sphere $S^{n-1}$ is $2\pi^{\frac{n} {2}}/\Gamma(\frac{n}{2})$, and letting $\nu$ be the corresponding measure so that
\[
\int_{S^{n-1}}\,d\nu = \frac{2\pi^{\frac{n} {2}}}{\Gamma(\frac{n}{2})},\ \ {\rm let} \ \ p_i=\int_{U_i}\,d\nu/\int_{S^{n-1}}\,d\nu.
\]
\begin{definition}
The {\em quantization output alphabet entropy} is given by the formula $-\sum_{i=1}^K p_i\log_2 p_i$.
\end{definition}
\smallskip
\begin{definition}
With $\alpha_i=$ the cardinality of the set of activation patterns mapped by $B$ to $d_i$ as above, let $q_i=\frac{1}{\alpha_i}p_i$.  The {\em activation pattern entropy} is then given by 
\[
H({\cal V}) =  
-\sum_{\tiny\begin{array}{c}
i\in[K]\\
v\in D^{-1}(d_i)\end{array}}
q_i \log_2 q_i.
\]
\end{definition}
\smallskip
\begin{proposition}
The following are equivalent formulas for activation pattern entropy:
$H({\cal V})=-\sum_{i=1}^K\alpha_iq_i\log_2 q_i = -\sum_{i=1}^K p_i\log_2 \frac{p_i}{\alpha_i}$.
\end{proposition}
\begin{proof}
The summation in the definition indicates sums must be taken over what we assume are equally probable activation patterns that correspond to direction $d_i$---which is to say there are $\alpha_i$ copies of $q_i$ in the sum.
\end{proof}

\smallskip

We note that the activation pattern entropy accounts for the information needed to specify which of a multiplicity of activation patterns is selected to specify the members of the quantization output alphabet.  It is always greater than or equal to the quantization output alphabet entropy.
\smallskip

\begin{example}
We return to $B$ in (\ref{eq:jb:simpleB}).  As noted, there are 81 activation patterns (elements of $\{-1,0,1\}^4$) but only 24 nonzero elements in the output alphabet;  this redundancy is characterized in detail in Fig.\ \ref{fig:jb:activate}.   We consider several approaches to the problem of emulating the linear vector field $Hx$ on the unit circle, with $H=
\left(\begin{array}{cc}
0 & 1\\
-1 & -2\end{array}\right)$.  
While direct combinatorial optimization is perhaps the most straightforward approach, we choose instead to apply learning algorithms {\bf Algorithm 1} given in Lemma \ref{lemma:jb:algorithm} and {\bf Algorithm 2} given in the Appendix.  This approach is aligned with our aim of exploring ideas suggested by principles of neuroscience.  The approach also affords the opportunity to examine the complexity of online {\em relearning} that may be needed to compensate for a channel dropout.  {\bf Algorithm 1} is especially efficient using the metric  (\ref{eq:jb:composite}) with various choices of parameters $h$ and $wt$, as $x$ ranges over the unit circle.   For the sake of simplicity of illustration, we take the simple loss function given by  the norm of the difference between $e^{Hh}x_0$ and $x_0+Bu$, and this is perhaps better suited to {\bf Algorithm 2}.  Note that $Hx$ is nowhere equal to zero for $x\in S^1$, and thus from the set of 81 possible activation patterns associated with matrix $B$ in (\ref{eq:jb:simpleB}), all nine that are mapped to $0$ by $B$ can be eliminated from consideration.  From the remaining set of 72 activation patterns, the only ones that turn out to be minimizers for the loss function on $S^1$ are a subset of 42.  If these activation patterns were i.i.d., the activation pattern entropy would be $log_2{42}\sim 5.39$, but given the way they arise in the learning protocol, they are not i.i.d.  Indeed, they provide optimal discrete approximations as $x$ traverses $S^1$ as described by the selection function defined in Table I.
The algorithm identifies a sub-alphabet of ten (10) output directions that  are used in the emulation.  These, together with their activation pattern multiplicities are listed in the table.  While the alphabet of 72 possible activation patterns contains a subset of 30 that are never used in this example, even the 42 that correspond to optimizers of the emulation metric has redundancies that are seen in the alphabet entropies.  For each direction, $d_i$, $\alpha_i$ is the multiplicity of corresponding activation patterns, and $p_i$ is the proportion of points on the circle where each direction minimizes the loss function.  The quantization output alphabet entropy and activation pattern entropy are thus respectively  
\[
-\sum_{k=1}^{10}p_k\log_2 p_k =  3.052 \ {\rm and} \ -\sum_{k=1}^{10}p_k\log_2 \frac{p_k}{\alpha_k} = 4.746.
\]

\medskip

\begin{table*}[h]
\begin{center}
\tiny
\caption{}
\begin{tabular}{||c||l|l|l|l|l||}
\hline 
$\theta$ &  $-0.321755 \le\theta<0.27167$ & $ 0.27167 < \theta < \frac{\pi}{6} $ & $ \frac{\pi}{6} < \theta <1.94263 $ & $ 1.94263 < \theta < 2.45243 $ & $ 2.45243 < \theta < 2.81984 $  \\ \hline
$f(\theta) $ & $\left(\begin{array}{c} 0 \\ -1 \end{array}\right)$\  $\begin{array}{l} \alpha_1=6 \\ p_1 = 0.094 \end{array}$ & $\left(\begin{array}{c} 0 \\ -2 \end{array}\right)$ \  $\begin{array}{l}\alpha_2 = 3 \\ p_2= 0.04\end{array}$ & $\left(\begin{array}{c} 1 \\ -2 \end{array}\right)$\ \ $\begin{array}{l}\alpha_3 = 2 \\ p_3=0.226 \end{array}$ & $\left(\begin{array}{c} 1 \\ -1 \end{array}\right)$ \  $\begin{array}{l} \alpha_4=4 \\ p_4= 0.08 \end{array}$ & $\left(\begin{array}{c} 1 \\ 0 \end{array}\right)$ \ $\begin{array}{l} \alpha_5=6 \\ p_5 = 0.06\end{array}$ \\ \hline\hline
& $2.81984 < \theta < 3.41326$ & $3.41326 < \theta < 7\pi/6$ &  $7\pi/ 6 < \theta <5.08422 $ & $5.08422 < \theta <5.59402$ & $ 5.59402 < \theta < 5.96143 $ \\ \hline
& $\left(\begin{array}{c} 0 \\ 1 \end{array}\right)$ \ $\begin{array}{l} \alpha_6=6 \\ p_6= 0.094 \end{array}$  & $\left(\begin{array}{c} 0 \\ 2 \end{array}\right)$\  $ \begin{array}{l} \alpha_7=3 \\ p_7 = 0.04 \end{array} $ & $\left(\begin{array}{c} -1 \\ 2 \end{array}\right)$\  $ \begin{array}{l} \alpha_8=2 \\ p_8= 0.226 \end{array} $ & $\left(\begin{array}{c} -1 \\ 1 \end{array}\right)$\  $ \begin{array}{l} \alpha_9=4 \\ p_9= 0.08 \end{array}$ & $\left(\begin{array}{c} -1 \\ 0 \end{array}\right)$\  $ \begin{array}{l} \alpha_{10}=6 \\ p_{10}= 0.06 \end{array}  $  \\ \hline
\end{tabular}
\label{table:jb:one}
\end{center}
\end{table*}

The example and  particular vector field we have chosen to illustrate the approach point to several general considerations in automated learning of quantization alphabets.  For both this example and one that uses the same activation without the activation redundancy (e.g.\ 
$B=\left(\begin{array}{cccc} 1&0&1/\sqrt{2}&1/\sqrt{2} \\
 0&1&1/\sqrt{2}&-1/\sqrt{2}\end{array}\right)$),
the useful activation patterns are considerably smaller than the available sets of patterns or output directions.  In the case of Example 1, the optimized output direction alphabet is such that only ten activation patterns are needed.  (See Table 1 and Fig.\ 2(c).)  Nevertheless to preserve resilience to channel dropouts, having a redundant activation pattern dictionary is useful---as illustrated in Figure 3.  The learning algorithms allow relatively fast online relearning of optimal emulation when any of the four channels drops out.

\begin{figure}[h]
	\begin{center}
		\includegraphics[scale=0.32]{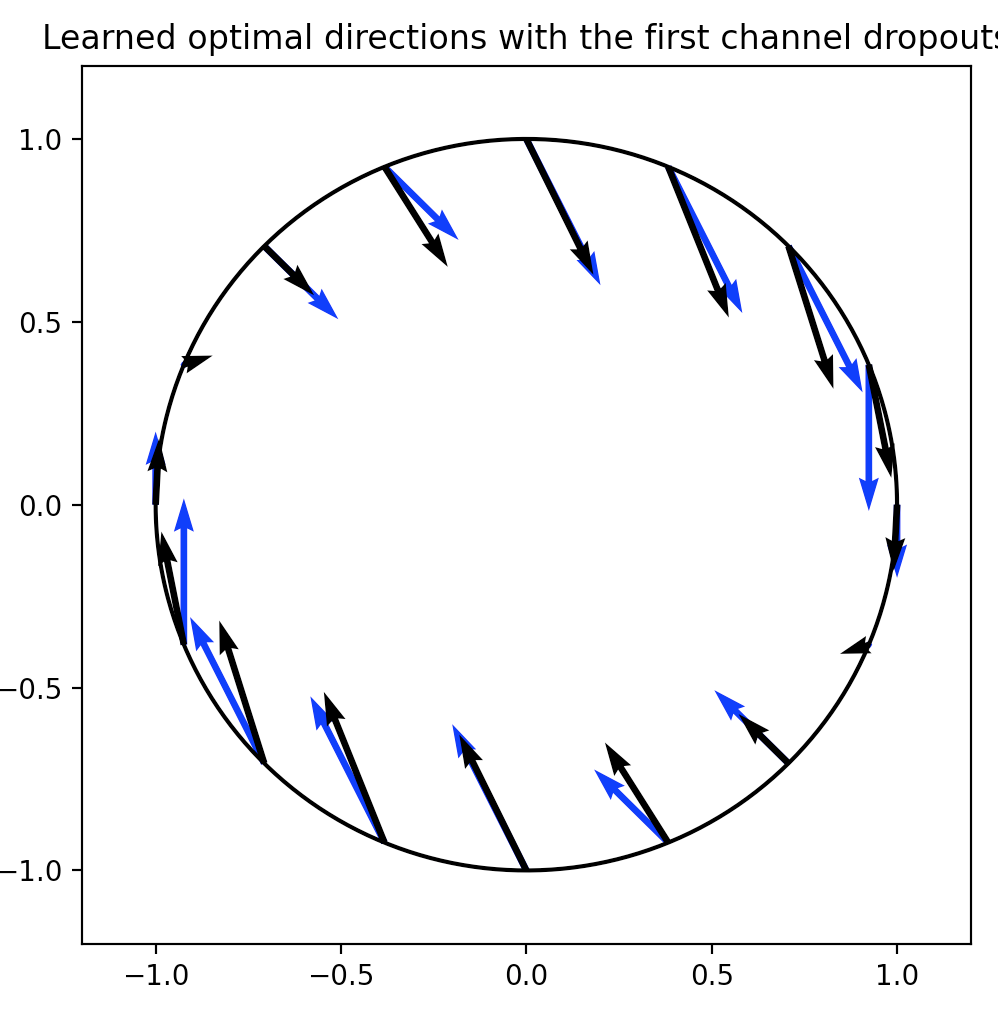}
		\includegraphics[scale=0.322]{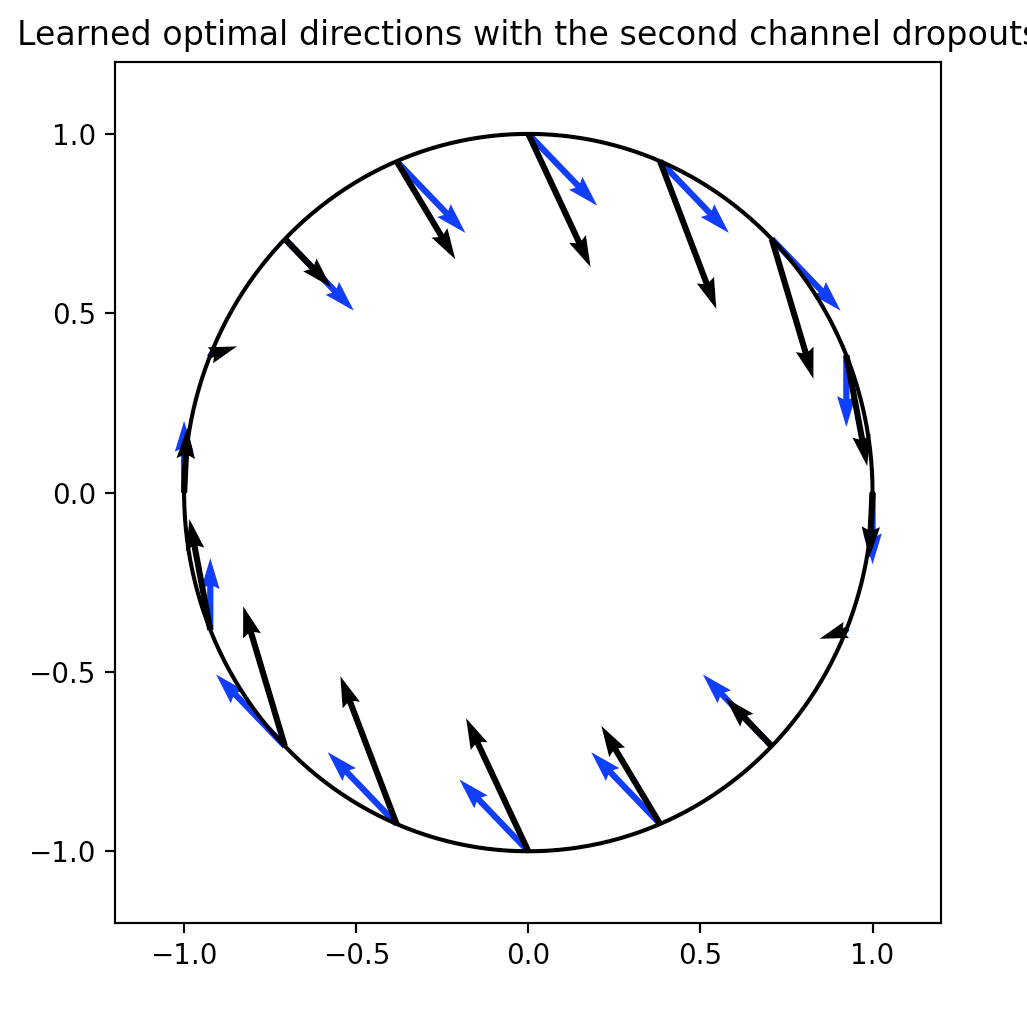}
		\includegraphics[scale=0.315]{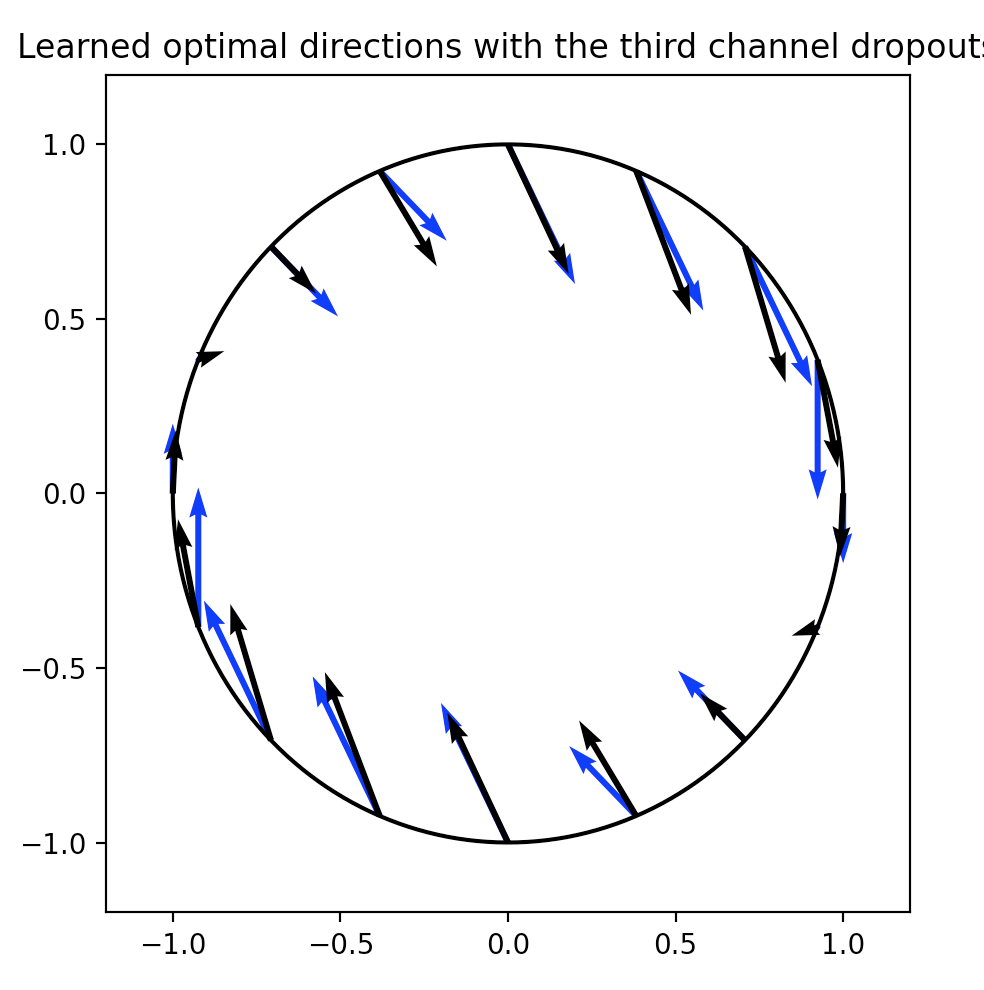}
		\includegraphics[scale=0.32]{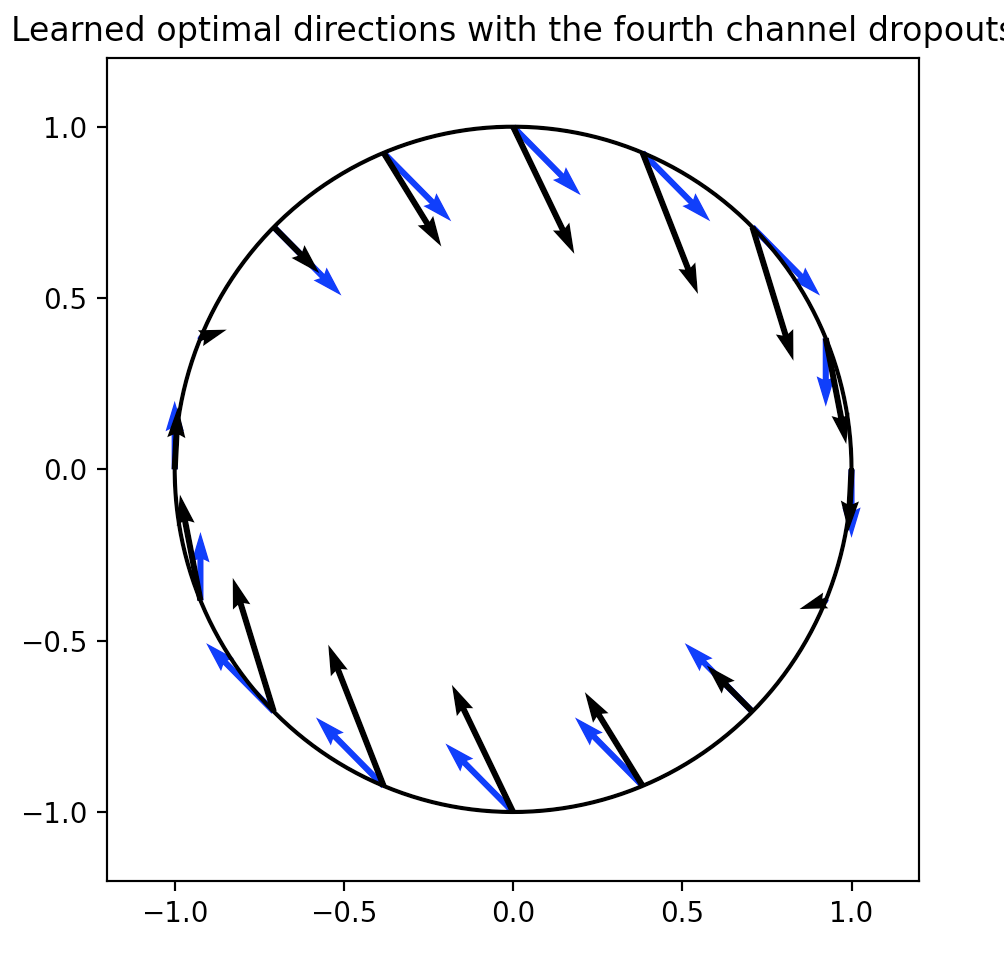}
		\includegraphics[scale=0.335]{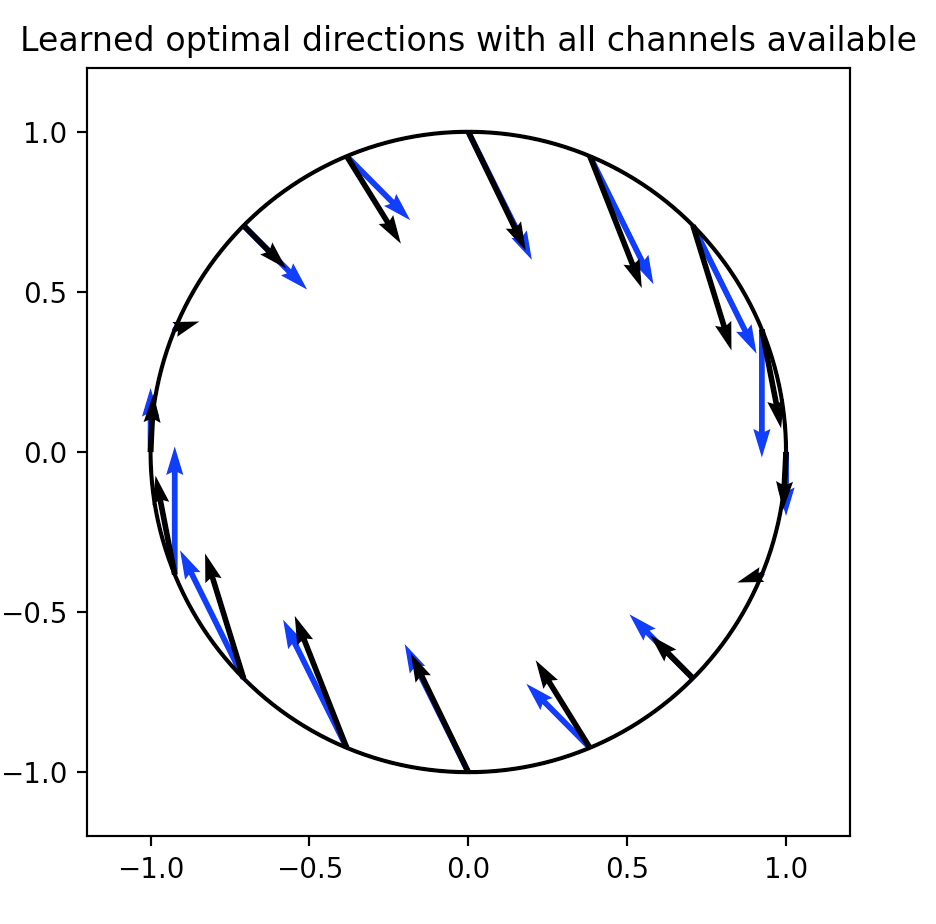}
	\end{center}
	\caption{For the emulation problem of Example 1, These figures show the learned optimal directions change when different channels become unavailable. From left-to right, the emulating direction vectors depicted are learned with respectively  channels 1,2,3 and 4 being unavailable and the bottom right image depicting learning with all channels operational.  Surprisingly, with channels 1 or 3 not available, the directions learned are the same as with all channels present.
}	\label{fig:ex}
\end{figure}

\end{example}

\smallskip


\bigskip

\section{Conclusions and future work}

The aim of the paper has been to describe the tools and some key concepts emerging from our ongoing investigation of control system models having very large numbers of input and output channels.  The idea of quantization using dictionaries of activation patterns and related quantization output alphabets has been introduced.  After how the complexity of such models scales with dimensions of the inputs, outputs, and state, a low dimensional example shows that these dictionaries and alphabets will typically exhibit considerable redundancy---not unlike written English and the Roman alphabet \cite{shannon1951prediction}.  Our approach has made use of the fact that autonomous (time-invariant) linear ordinary differential equations have a certain length-scale invariance, thus key features can be studied restricting attention to the surface of the unit sphere.  Part of future work will be to understand similar principled quantization for nonlinear systems.  Even in the linear case, however, there are interesting open questions that are illustrated by the examples.  The trajectories of the vector field in Fig.\ 2(c) approach the origin through ever narrower channels (one of which is confined to the dark-bordered pie wedge).  It seems likely that effective quantization for these systems will need different levels of coarseness as the asymptotic limits are appraoched.  We conjecture that a theory of trajectory-aware quantization---where there is refinement as time passes---may benefit from recent work of Molloy and Nair on what they call ``smoother entroy'' \cite{Nair21}.  Finally, a caveat.  The methods and approach we have described are  not (yet at least) being proposed for actual applications.  The aim is simply to explore information processing concepts.






\begin{appendix}

\begin{breakablealgorithm}
	
		\caption{Learning activation patterns for linear system vector field}\label{alg:cap2}
		\begin{algorithmic}[1]
		\State Input: Activation patterns $U=\{u_1,u_2,...,u_K\}$ and its direction vectors of quantization output alphabet to form an action space $Dir=\{d_1,d_2,...,d_K\}$, a starting point $x_0$, a learning metric $G(u_k,x_0)$;
		\State Initialize replay memory $D$ with capacity $N$;
		\State Initialize action-value $Q$ function with parameter $\phi$  and target-$Q$ function with parameter $\hat\phi=\phi $, threshold $\kappa= G(u_k,x_0)$ for an arbitrary activation pattern $u_k$ whose direction  is $d_k=Bu_k$;
		\While{True}
		\Repeat
		\State Start from the initial point $x_0$ for the quantized system;
		\Repeat 
		\State Choose direction $d_t=\pi^\epsilon$, get reward $r_t=-G(u_t,x_0)$ and new state $x_{t+1}$ at time t;
		\State Store $(x_t, d_t, r_t,x_{t+1})$ as a memory cue to $D$;
		\State Sample random minibatch of cues  $(x_j, d_j, r_j,x_{j+1}),j\in[0,t]$ from $D$;
		\State $y=\Big\{
		\begin{array}{ll}
			r_j, \ \ \ if \ x_{j+1}\ is\ a\ terminal\ state &\\
			r_j+\gamma max_{d'}\ Q_{\hat\phi}( x_{j+1},d') \ \ otherwise&
		\end{array}$
		\State Applying loss function $(y- Q_{\phi}( x_j,d_j)^2$ to train $Q$ network;
		\State $x_t\gets  x_{t+1}$;
		\State Every $T$ steps,  $\hat\phi\gets\phi$;
		\Until{$G(u_t,x_0)<\kappa$}
		\Until{for the given $x_0$, form convergent paths}
		\If {the convergent direction sequence contains more than one direction}			
		\State	terminate the while loop
		\ElsIf{the convergent direction sequence contains exactly one direction $d$ generated by $u$}
		\State $\kappa=G(u,x_0)$;
		\State $d_0=d$;
		\EndIf	
		\EndWhile
		
		\State Output: Optimal direction $d_0$ and its correspondence activation pattern ${u_{opt}}$.
	\end{algorithmic}
	
\end{breakablealgorithm}

\end{appendix}

\bibliography{references}
\bibliographystyle{IEEEtran}

\end{document}